\documentclass[conference]{IEEEtran}
\IEEEoverridecommandlockouts
\usepackage{cite}
\usepackage{amsmath,amssymb,amsfonts}
\usepackage{amsthm}
\usepackage{float}

\newtheorem{problem}{Problem}
\newtheorem{defn}{Definition}
\newtheorem{prop}{Proposition}

\usepackage{algorithmic}
\usepackage{graphicx}
\usepackage{textcomp}
\usepackage{xcolor}

\usepackage{cleveref}
\usepackage{microtype}

\def\BibTeX{{\rm B\kern-.05em{\sc i\kern-.025em b}\kern-.08em
    T\kern-.1667em\lower.7ex\hbox{E}\kern-.125emX}}

\begin{document}

\title{Barron-Wiener-Laguerre models}

\author{
\IEEEauthorblockN{Rahul Manavalan}
\IEEEauthorblockA{\textit{Center for Mathematical Sciences} \\  \textit{Lund University}}
\and
\IEEEauthorblockN{Filip Tronarp}
\IEEEauthorblockA{\textit{Center for Mathematical Sciences} \\ \textit{Lund University}}
}

\maketitle

\begin{abstract}
We propose a probabilistic extension of Wiener–Laguerre models for causal operator learning. Classical Wiener–Laguerre models parameterize stable linear dynamics using orthonormal Laguerre bases and apply a static nonlinear map to the resulting features. While structurally efficient and interpretable, they provide only deterministic point estimates.
We reinterpret the nonlinear component through the lens of Barron function approximation, viewing two-layer networks, random Fourier features, and extreme learning machines as discretizations of integral representations over parameter measures. This perspective naturally admits Bayesian inference on the nonlinear map and yields posterior predictive uncertainty.
By combining Laguerre-parameterized causal dynamics with probabilistic Barron-type nonlinear approximators, we obtain a structured yet expressive class of causal operators equipped with uncertainty quantification. The resulting framework bridges classical system identification and modern measure-based function approximation, providing a principled approach to time-series modeling and nonlinear systems identification.
\end{abstract}


\section{Introduction}

\subsection{Problem setup}
We are interested in solving causal operator learning problems such as timeseries modeling and systems identification. These problems are briefly presented for completeness.

\begin{problem}(Timeseries modeling) 
Let $T, T_{max} \in \mathbb{R}^+$,
$\mathcal{T} := [0,T] \cup (T, T_{max}]$, 
$u \in L^2([0,T_{max}], \mathbb{R}^n)$.
Let $\hat{\mathcal{T}} := \{t_0, t_1,\dots,t_m\} \subset [0,T]$ such that $t_0 = 0, t_M = T$. 
Given $\hat{U}:=\{u(t_0) + \epsilon_0, u(t_1) + \epsilon_1, \dots, u(t_{M}) + \epsilon_M\}$ where $\epsilon_i \sim \mathcal{N}(0, \sigma I)$ and $\sigma$ is unknown observation noise parameter, estimate $\tilde{u} \in L^2([0,T_{max},\mathbb{R}^n])$ such that for a constant $\eta >0$, $\|u-\tilde{u}\|_{L^2} < \eta$. 
\end{problem}

\begin{problem}(Systems-identification)
    Let $T, T_{max} \in \mathbb{R}^+$, 
    $\mathcal{T}:=[0,T] \cup (T, T_{max}]$.
    Let $a \in L^2([0,T_{max}], \mathbb{R}^{in})$, 
    $u \in L^2([0,T_{max}], \mathbb{R}^{out})$. The operator $J : L^2_{in} \mapsto L^2_{out}$ is bounded, continuous. 
    For $\hat{\mathcal{T}}:=\{t_0, t_1, \dots, t_M\}$ such that $t_0 = 0, t_M = T$, given the observations $\hat{A}:=\{a(t_0)+\epsilon_0,a(t_1)+\epsilon_1, \dots, a(t_M)+\epsilon_M \}$ and $\hat{U}:=\{u(t_0) + \beta_0, u(t_1) + \beta_1, \dots, u(t_M) + \beta_M\}$ where $\epsilon_i \sim \mathcal{N}(0,\sigma I)$ and $\beta_i \sim \mathcal{N}(0,\nu I)$ for observation noise parameters $\sigma, \nu > 0$, find a mapping $\tilde{J} : L_{in}^2 \mapsto L_{out}^2$ for $\delta > 0$, such that $\|\tilde{J} a - u\|_{L^2_{out}} < \delta.$
\end{problem}
\subsection{Previous Work}

The problem of learning nonlinear causal operators is classical. 
Volterra--Wiener theory provides universal expansions for fading-memory systems \cite{Wiener1958,Schetzen1980}. 
While expressive, these representations suffer from combinatorial parameter growth with memory depth. 
Boyd and Chua formalized fading-memory approximation and clarified the need for structured subclasses \cite{BoydChua1985}.
Block-oriented models, such as Wiener and Hammerstein systems, restrict the operator to a stable linear time-invariant (LTI) system composed with a static nonlinearity. 
This reduction is theoretically justified and practically dominant in nonlinear system identification \cite{9721138, LI2024114265, BAI2009736}. 
Further structural compression is achieved through orthonormal basis parameterizations of stable LTI systems. 
The generalized orthonormal basis function framework established consistent identification theory for such models \cite{VandenHof1995,Heuberger1995}. 
Wiener--Laguerre models inherit interpretability and computational tractability and remain central in systems identification \cite{Hagenblad2008}. 
However, classical formulations are deterministic and do not provide uncertainty quantification.

In parallel, approximation theory evolved along a different axis. 
Two-layer neural networks were shown to be universal approximators \cite{Cybenko1989,Hornik1989,Hornik1991}. 
Barron refined this picture by identifying a Fourier-defined function class for which shallow networks achieve dimension-independent convergence rates \cite{Barron1993}. 
Subsequent work connected these results to greedy approximation and convex formulations \cite{Barron2008,Bach2017ConvexNN}. 
These developments demonstrate that shallow models can avoid the curse of dimensionality under structural assumptions on the target.
A more structural reinterpretation emerged through measure-based formulations of two-layer networks. 
In this view, a neural network is an integral representation over parameter space. 
Training can be interpreted as gradient flow in probability space. 
The mean-field program of Weinan E and collaborators established population-risk bounds and clarified global convergence in overparameterized regimes \cite{E2019Risk,E2021Barron,ChizatBach2018,Mei2018,Rotskoff2018}. 
This framework connects neural networks, kernel methods, and optimal transport.
Random feature methods arise as Monte Carlo discretizations of these integral models. 
Random Fourier Features approximate shift-invariant kernels via spectral sampling \cite{RahimiRecht2007,RahimiRecht2008}, 
with subsequent analyses relating convergence to spectral decay and effective dimension \cite{Le2013Fastfood,Bach2017Quadrature,RudiRosasco2017}. 
Extreme Learning Machines fix hidden parameters randomly and solve only for linear output weights while retaining universal approximation guarantees \cite{Huang2006ELM,Huang2012ELM}.

These two lines of work: structured dynamical operator modeling via orthonormal bases \cite{Wiener1958,Schetzen1980,BoydChua1985,VandenHof1995,Heuberger1995,Hagenblad2008}, 
and measure-theoretic shallow approximation with dimension-robust guarantees \cite{Barron1993,E2021Barron,RahimiRecht2007}  have largely evolved independently. 
Their systematic integration—combining structured causal dynamics with probabilistic Barron-type nonlinear approximation—remains largely unexplored.

\section{Wiener-Laguerre models}
Consider the canonical linear model for a scalar valued input function $u$. 
\begin{align}
    \begin{gathered}
        \frac{d^px}{dt^p}(t) = u(t). 
    \end{gathered}
\end{align}
We may assume that the initial conditions of $x$ and up to its $p-1^{th}$ derivative are zero. Here, $p$ referred to the order of the model.  This linear system can be equivalently denoted by:
\begin{align}
    \begin{gathered}
        \frac{dy}{dt}(t) = A y(t) + B u(t), \quad y(0) = 0. 
    \end{gathered}
\end{align}
Its solution is formally given in  terms of the impulse response:
\begin{align}
    \begin{gathered}
        y(t) = \int_0^t d\tau \; \exp^{A(t-\tau)} B u(\tau ) := \int_0^t d\tau  k(t-\tau)  u(\tau). 
    \end{gathered}
\end{align}
The impulse response for this linear model is given by:
\begin{align}
    \begin{gathered}
        k(t-\tau) := \begin{bmatrix}
            \frac{(t-\tau )^{p-1}}{(p-1)!}  \; \frac{(t-\tau )^{p-2}}{(p-2)!}  \; \dots \; \frac{(t-\tau )}{1} \; 1
        \end{bmatrix}^T. 
    \end{gathered}
\end{align}
This model has infinite memory. Therefore, for longer time intervals, the resulting features $y$ are generally polynomially growing in $t$. To resolve this, it is necessary to incorporate forgetfulness into the linear model. We will realize this by parameterizing the impulse response using exponentially decaying Laguerre polynomials, also known as Laguerre functions. 
In the rest of this section, we will derive this basis. 

\begin{defn}{\textbf{Rescaled Laguerre polynomials}} \\ 
    Let $L_m$ and $L_n$ represent the $m^{th}$ and $n^{th}$ Laguerre polynomials respectively. The rescaled Laguerre polynomials may be derived from the orthogonality condition by a change of variable:
    \begin{align}
        \begin{gathered}
            \int_0^{\infty} dt \; e^{-t} \; L_m(t) \; L_n(t) = \delta_{m,n}. \\ 
            \equiv \int_0^{\infty} ds \; (\sqrt{2 \lambda} \;e^{-\lambda s} \; L_m(2\lambda s)) \; (\sqrt{2 \lambda} \;e^{-\lambda s} \; L_n(2\lambda s)), \\ 
            := \int_0^{\infty} ds \; l_m(s) \; l_n(s) = \delta_{m,n}.
        \end{gathered}
    \end{align}
\end{defn}

Here $t:=2\lambda s$.
The corresponding state matrices $A^L,B^L$ for the rescaled Laguerre bases may be derived using the following:
\begin{prop}[\cite{szego75}]
    \begin{align}
        \begin{gathered}
            \frac{d}{dt} L_m(2 \lambda t) = \frac{d}{dt} L_{m-1}(2 \lambda t) -2 \lambda  L_{m-1}(2 \lambda t).
        \end{gathered}
    \end{align}
\end{prop}
Its corollary is:
    \begin{align}
        \begin{gathered}
            \frac{d}{dt} l_m(t) = \frac{d}{dt} l_{m-1}(t) -\lambda  \left[l_{m-1}(t) + l_m(t)\right].
        \end{gathered}
    \end{align}
    \label{eq:rescaled}
\begin{prop}
    When the impulse response of a linear autonomous dynamical system is parameterized using rescaled-Laguerre polynomials, that is $\kappa^{RL}$, the state, input matrices are:
    \begin{align}
        \begin{gathered}
            A^L := \begin{bmatrix}
                -\lambda & -2\lambda & -2\lambda & \dots & -2\lambda \\ 
                &  - \lambda & -2\lambda & \dots & -2\lambda \\ 
                & & - \lambda  & \dots & -2\lambda \\ 
                & & & \ddots & \vdots \\ 
                & & & & - \lambda 
            \end{bmatrix}, \\  \quad B^L = \sqrt{2 \lambda} \begin{bmatrix}
                1 & 1 & \dots 1 & 1
            \end{bmatrix}^T.
        \end{gathered}
    \end{align}
\end{prop}

\begin{proof}[Sketch of proof:]
    From the corollary of proposition 1, we may obtain a linear relationship between the derivative of a vector of Laguerre functions and the vector of functions itself. Noticing that this resembles the augmented ODE with no input function, we get $A^L$. Furthermore, note that $B^L = k(0).$
\end{proof}

\begin{defn}{\textbf{Wiener-Laguerre model}} \\
    For an instance of the input $u \in \mathcal{U}$ and the output $z \in \mathcal{Z}$, the Wiener-Laguerre model characterizes their relation as:
    \begin{align}
        \begin{gathered}
            z(t) = \phi^{\theta}(y(t)), \\ 
            y(t) := \int_0^t d\tau \; A^Ly(\tau) +  B^L u(\tau), \\ 
            \phi^{\theta}(y)(t) := \sum_{i=1}^k a_i \sigma\left(W_i^T y(t) + b_i\right).
        \end{gathered}
    \end{align}
\end{defn}

The same model can be extended to vector-valued input functions by stacking these models on top of one another. For a sequence of model order $\{p_i\}_{i=1}^{d_{in}} \subset \mathbb{N}$, we may define an ensemble of linear ODEs:
\begin{align}
    \begin{gathered}
        \frac{dy_1}{dt}(t) = A_1^L y_1(t) + B_1^L u_1(t), \\ 
        \vdots \\ 
        \frac{dy_{d_{in}}}{dt}(t) = A_{d_{in}}^L y_{d_{in}}(t) + B_{d_{in}}^L u_{d_{in}}(t),
    \end{gathered}
\end{align}
with all of their initial conditions $y_i(t) = 0$. Note that each $y_i(t) \in \mathbb{R}^{p_i}. $ Each $A_i^L$ may be associated with a different forget factor $\lambda_i$. Denoting this ensemble as:
\begin{align}
    \begin{gathered}
        \frac{dw}{dt}(t) = \hat{A} w(t) + \hat{B} u(t), \\ 
        w(t) = \int_0^T d\tau \; \hat{A} w(\tau) + \hat{B} u(\tau).
    \end{gathered}
\end{align}
The Wiener Laguerre model has a similar algebraic form, where the weights and biases have been appropriately redefined. 
\begin{align}
    \begin{gathered}
        z(t) = \phi_{\theta} (w(t)), \\ 
          \phi^{\theta}(w)(t) := \sum_{i=1}^k a_i \sigma\left(W_i^T w(t) + b_i\right).
    \end{gathered}
\end{align}

The hyper parameters of this model namely $\{p_i\}_{i=1}^{d_{in}}$, $\{\lambda_i\}_{i=1}^{d_{in}}$ are tuned on a case by case basis whereas the parameters of the neural network -- $\{W_i, b_i, a_i\}_{i=1}^k$ can be trained using back-propagation. We will use the mean squared error as the discrepancy function for a given input function.

\section{Probabilistic Barron functions}
A Wiener Laguerre model only offers point estimates of the output function. In practice, it will be of interest to consider models that can quantify the uncertainty in these predictions. In this section, we will derive such non-linear function approximation schemes  based on Barron functions. 

\begin{defn}[Barron function \cite{E2021Barron}]
Let $\sigma:\mathbb{R}\to\mathbb{R}$ be a measurable function. 
A function $f:\mathbb{R}^d \to \mathbb{R}$ is called a Barron function if there exists a finite signed measure 
$\mu$ on $\Theta := \mathbb{R} \times \mathbb{R}^d \times \mathbb{R}$ such that
\begin{equation}
f(x) = \int_{\Theta} a\, \sigma(W^\top x + b)\, d\mu(a,W,b).
\end{equation}
\end{defn}

Depending on the choice of $\sigma$, $\mu$; Barron functions reduce to different function approximation schemes. For instance, for arbitrary $\sigma$, if we parameterize $\mu$ using the atomic measure:
\begin{align}
    \begin{gathered}
        \mu_1(a,W,b) = \sum_{i=1}^K \delta(a,a_i) \; \delta(W, W_i) \; \delta(b, b_i), 
    \end{gathered}
\end{align}
then one arrives at a two-layer neural network. 
Similarly, when $\sigma = \cos $ and the weights $\{W_i\}_{i=1}^K$ and biases $\{b_i\}_{i=1}^K$ are sampled i.i.d such that:
\begin{align}
    \begin{gathered}
        W_i \sim \mathcal{N} \left(W; 0, \frac{1}{l^2} I\right), \quad 
        b_i \sim \mathcal{U}(b; [0, 2\pi]),
    \end{gathered}
\end{align}
and when the measure $\mu$ is parameterized as:
\begin{align}
    \begin{gathered}
        \mu_2(a,W,b) :=  \sum_{i=1}^K a_i \; \delta(W,W_i)\; \delta(b,b_i) ,
    \end{gathered}
\end{align}
this corresponds to the classical random Fourier feature \cite{RahimiRecht2007} approximation of a kernel machine estimator. Analogously, for arbitrary $\sigma$, when  we choose to use:
\begin{align}
    \begin{gathered}
        W_i \sim \mathcal{N} \left(W; 0, \frac{1}{d} I_d\right), \quad b_i \sim \mathcal{N}(b; 0, 1)
    \end{gathered}
\end{align}
and construct the measure $\mu_3$ in the same way as $\mu_2$, it corresponds to the extreme learning approximator \cite{Huang2006ELM}. In both these cases, we obtain generalized linear interpolants. Denoting by $\hat{\Omega} \subset \mathbb{R}^d$ where target function $g$ is observed, the feature matrix $\Phi \in \mathbb{R}^{K \times M}$ is calculated as $\Phi_{ij} = \sigma(W_i^T \hat{\Omega}_j + b_i)$. The estimation proceeds by solving the following least squares problem \footnote{Here, the symbol $A_{:,i}$ means that the $i^{th}$ column of a matrix A is indexed.}:
\begin{align}
    \begin{gathered}
        a^* := \arg \min_{a \in \mathbb{R}^{n \times K}} \sum_{i=1}^M \| a \Phi_{:,i}  - g(\hat{\Omega}_i)\|^2_2
    \end{gathered}
\end{align}

\subsection{Bayesian estimation}
In practice, observations are noisy.
Here, we assume a Gaussian noise model:
\begin{align}
g(\hat{\Omega}) &= a \Phi  + \varepsilon, \quad
\varepsilon \sim \mathcal{N}(0,\sigma^2 I_M).
\end{align}

For $\alpha > 0$, a Gaussian prior is placed on the coefficients:
\begin{align}
a \sim \mathcal{N}(0,\alpha^{-1} I_K).
\end{align}

By conjugacy, the posterior distribution of $a$ is Gaussian:
\begin{align}
\begin{gathered}
a \mid g(\hat{\Omega}) \sim \mathcal{N}(m_*, \Sigma_*), \\
\Sigma_*
=
\left(
\alpha I_K
+
\frac{1}{\sigma^2}
\Phi^T \Phi
\right)^{-1}, \quad
m_*
=
\frac{1}{\sigma^2}
\Sigma_* \Phi^T y.
\end{gathered}
\end{align}

For a new input $x_* \in \mathbb{R}^d$, let $\phi_* := \sigma(\tilde{W}^Tx_* + \tilde{b})$. 
Here $\tilde{W}$ and $\tilde{b}$ are the pre-sampled weights and biases. 
The posterior predictive distribution of the latent function value is
\begin{align}
f_* \mid g(\hat{\Omega})
\sim
\mathcal{N}
\left(
\phi_*^T m_*,
\;
\phi_*^T \Sigma_* \phi_*
\right).
\end{align}

\subsection{Numerical verification}

 We will compare these estimators on a trimodal Gaussian mixture function. The advantage of this test, is that, the function is known a-priori and that it scales to arbitrary dimensions. In this work, we will be primarily interested in measurements of functions up to 5 input dimensions. For the domain $\Omega := [-30, 30]^d$. For three points $x_1 := -15 \; 1_d$ , $x_2 := 15 \; 1_d$, $x_3 := 0_d$, the target function is defined as:
    \begin{align}
        \begin{gathered}
            f(x) = \sum_{i=1}^3 \mathcal{N}(x; x_i, 20 I_d). 
        \end{gathered}
    \end{align}



\begin{table}[ht]
\centering
\begin{tabular}{c|c|c||c|c}
\hline
 & \multicolumn{4}{c}{Models} \\
\cline{2-5}
$d$ 
& \multicolumn{2}{c||}{RFF} 
& \multicolumn{2}{c}{ELM} \\
\cline{2-3}\cline{4-5}
 & Mean & Std & Mean & Std \\
\hline
1  & $1.065\times 10^{-7}$ & $1.616\times 10^{-4}$
   & $2.148\times 10^{-7}$ & $2.080\times 10^{-4}$ \\
2  & $6.341\times 10^{-6}$ & $4.865\times 10^{-3}$
   & $6.168\times 10^{-6}$ & $4.210\times 10^{-3}$ \\
3  & $4.018\times 10^{-1}$ & $4.466\times 10^{-3}$
   & $1.697\times 10^{-2}$ & $9.304\times 10^{-3}$ \\
4  & $2.382\times 10^{0}$  & $3.922\times 10^{-3}$
   & $7.872\times 10^{-1}$ & $4.131\times 10^{-3}$ \\
5  & $1.435\times 10^{0}$  & $3.886\times 10^{-3}$
   & $8.884\times 10^{-1}$ & $3.998\times 10^{-3}$ \\
\hline
\end{tabular}
\vspace{2mm}
\caption{Approximation error (mean and standard deviation) for the tri-modal Gaussian function using Barron functions for increasing input dimension $d$ (fixed sample size).}
\label{tab:gaussian}
\end{table}

\Cref{tab:gaussian} reports RFF and ELM performance on a tri-modal Gaussian mixture with fixed training samples ($M=2000$) and features ($K=1500$) across dimensions, probing a finite-budget regime.
For $d=1,2$, both methods are accurate; errors are small and variability stems mainly from random feature sampling rather than bias.
From $d \ge 3$, the mean error increases sharply while the standard deviation stays roughly constant. This points to structural under-resolution due to fixed sample and feature budgets, not instability.
Increasing both samples and features mitigates this effect.

\section{Numerical experiments}
In the following, we will demonstrate our methods on systems-identification and timeseries modeling tasks. 

\subsection{Systems identification}
The synthetic data for our systems identification task is generated as follows. 
The input signal is constructed as a finite Fourier series
\begin{align}
    \begin{gathered}
        u(t) = \sum_{k=1}^{5} \frac{1}{\sqrt{5}} 
\sin(k \omega_0 t + \phi_k),
    \end{gathered}
\end{align}
with $\omega_0 = 1.0$ and phases $\phi_k \sim \mathrm{Unif}[0,2\pi]$.
The output is generated by the forced second-order system:
$\ddot{y}(t) + 0.8\,\dot{y}(t) + 4.0\,y(t) = 1.2\,u(t),$
with $y(0)=0$, $\dot{y}(0)=0$.
The system is simulated on $[0,50]$ with step size $dt=0.01$, and Gaussian observation noise with standard deviation $\sigma=0.02$ is added to the output.
This data is modeled using the Bayesian random Fourier feature (RFF) 
Wiener--Laguerre model. 
The linear dynamical component is parameterized using $p = 15$ Laguerre features 
with forgetting factor $\lambda = 30.0$. 
The nonlinear map is approximated using $K = 1000$ random Fourier features. 
Bayesian linear regression is performed with regularization parameter 
$\sigma = 8 \times 10^{-2}$. The predictions from this model is show in \cref{fig:sys-id-experiment}. 
The RMSE metric of the difference between the reference and the model prediction in the test region is $\eta := 0.07620$. The mean variance across samples of this difference is $0.01519$. 

\begin{figure}
    \centering
    \includegraphics[width=\linewidth]{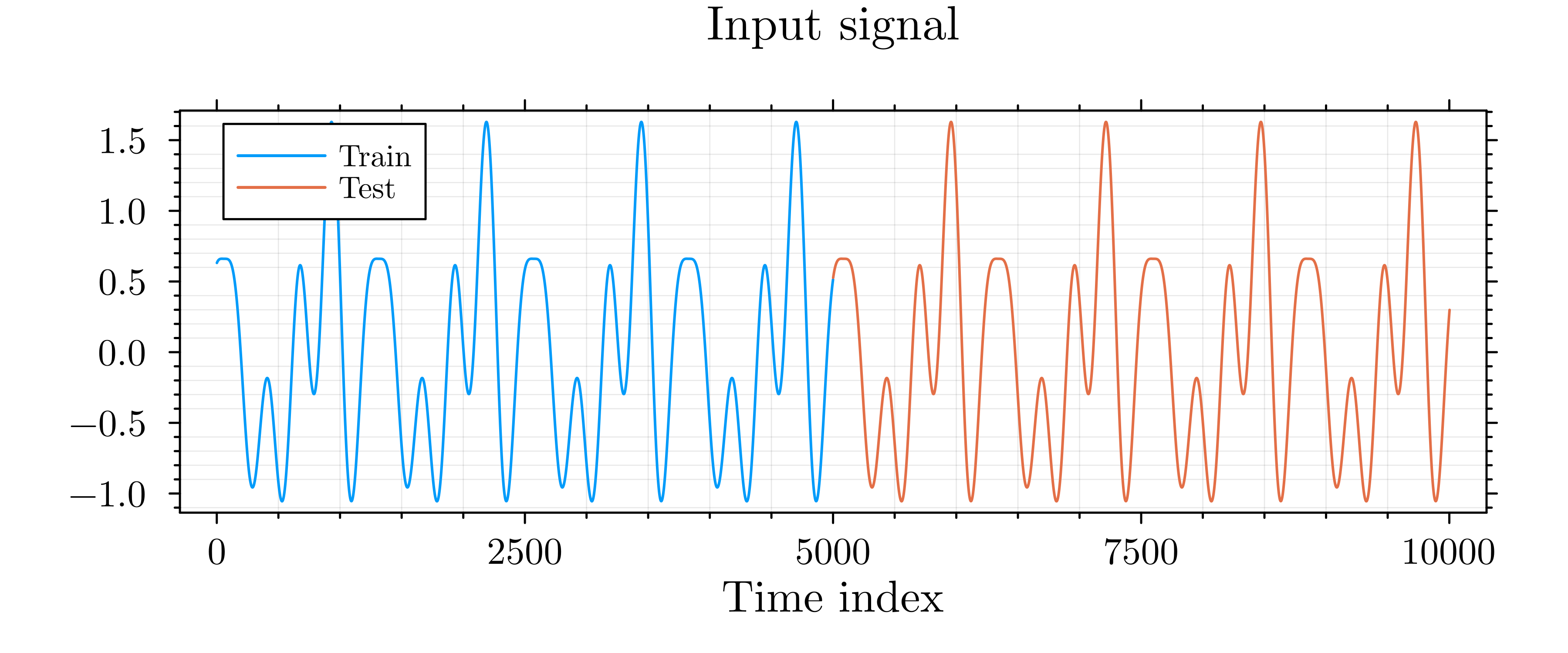}
    \includegraphics[width=\linewidth]{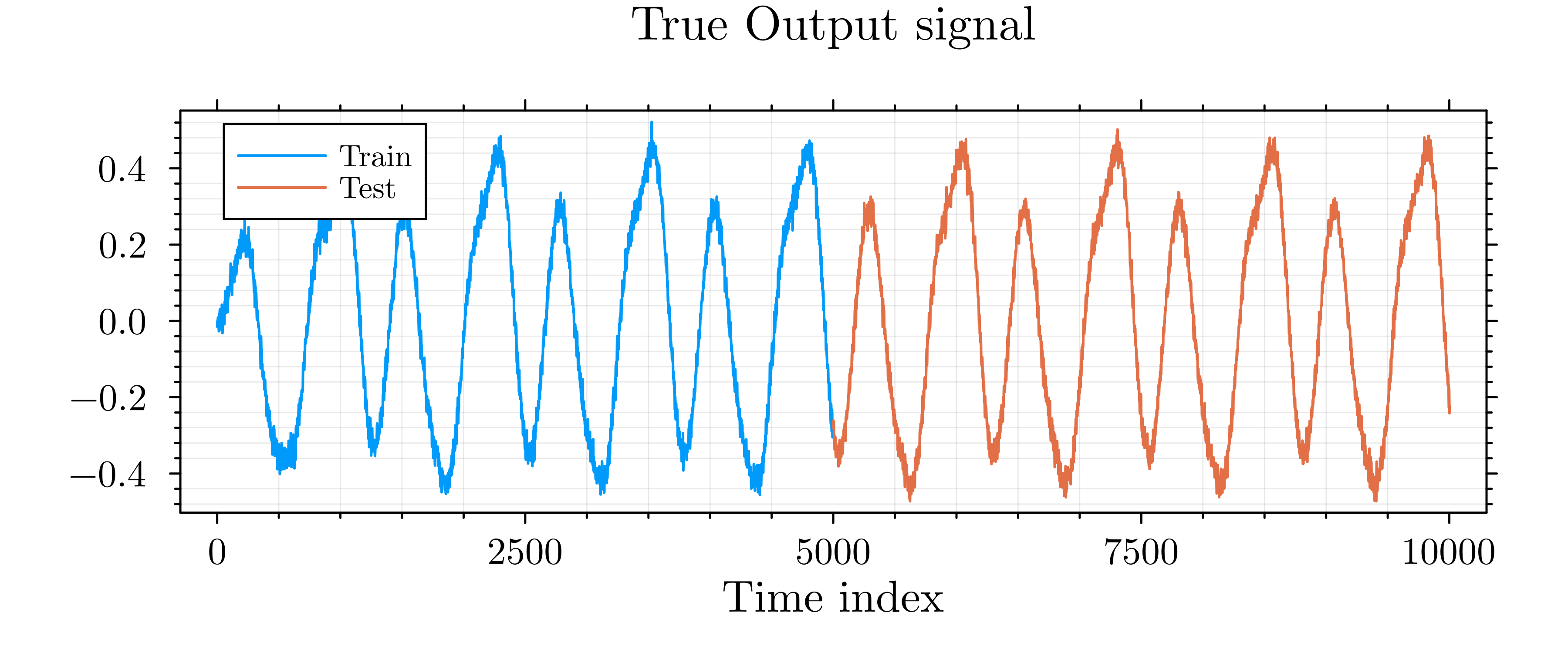}
    \includegraphics[width=\linewidth]{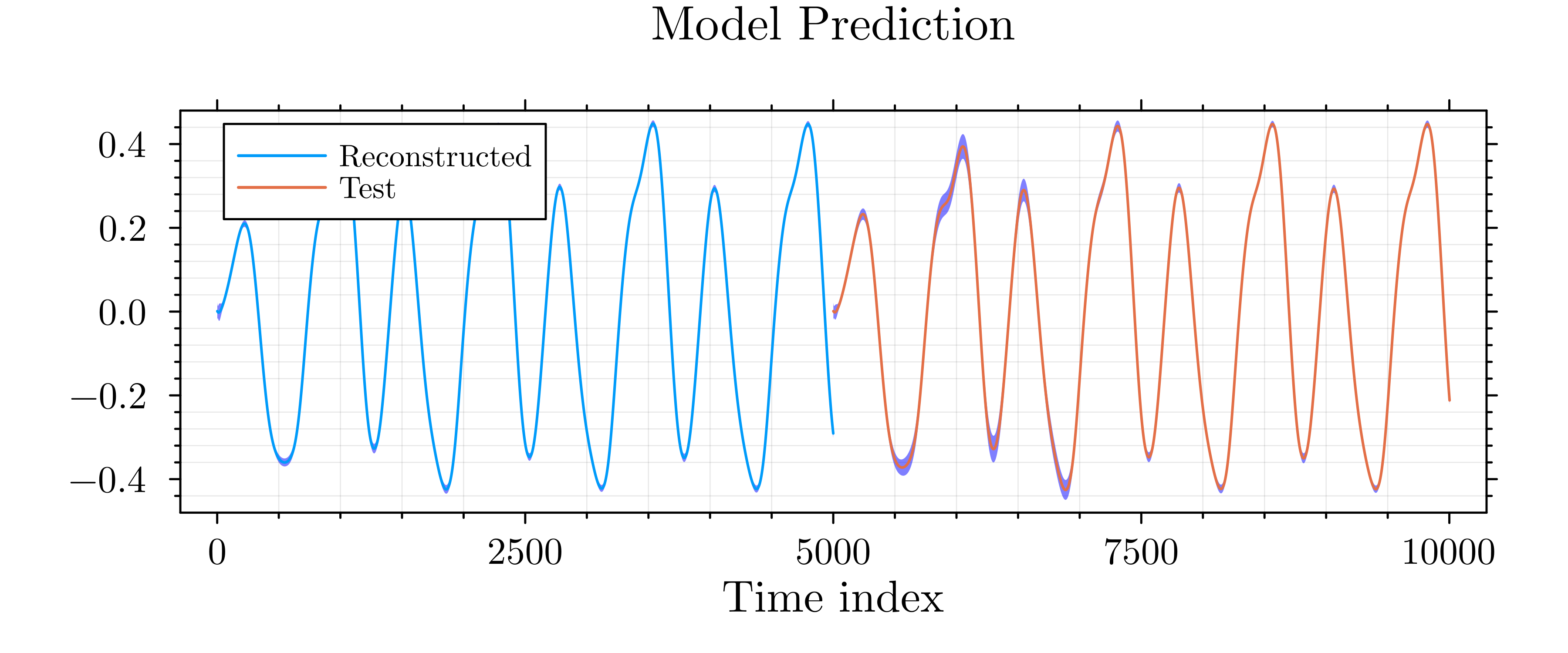}
    \caption{\textbf{Systems identification}: Train–test evaluation of a Bayesian RFF-Barron-Wiener–Laguerre model applied to a second order differential equation system. Here, time-index is the discrete sample number corresponding to uniformly sampled time points.
Top: Input signal with training (blue) and held-out test (orange) regions.
Middle: True system output corresponding to the full trajectory.
Bottom: Model reconstruction and prediction.
The model is trained on the first half of the trajectory and evaluated on the unseen second half, demonstrating accurate phase and amplitude tracking in the test region.}
    \label{fig:sys-id-experiment}
\end{figure}

\subsection{Timeseries modeling}

For the second task, data is generated from the nonlinear 
Van der Pol oscillator in the oscillatory regime,
\begin{equation}
\dot{x}(t) = v(t), 
\qquad
\dot{v}(t) = \mu \big(1 - x(t)^2\big) v(t) - x(t),
\end{equation}
with parameter $\mu = 2.0$ and initial condition 
$(x(0), v(0)) = (2.0, 0.0)$.
The system is simulated on $[0,40]$ with step size $dt = 0.01$. 
Additive Gaussian noise with standard deviation $\varepsilon = 10^{-1}$ 
is applied to both state components to obtain the observed time series.
Modeling this timeseries data can be posed as a systems identification problem where the timeseries is assumed to be the input signal and for a given point $t \in \mathbb{R}^+$, the $t+ k\delta t$ shifted timeseries as the output. Here, assume $k=1$ and $\delta t$ as the increment between two observation points. This modified dataset is learned using a Bayesian ELM--Barron--Wiener--Laguerre model. 
The linear dynamical component is parameterized using $p = 50$ Laguerre features 
with forgetting factor $\lambda = 3.0$. 
The nonlinear map is approximated using an extreme learning machine 
with $2000$ hidden neurons. 
Bayesian linear regression is performed with regularization parameter 
$\sigma = 5 \times 10^{-1}$. The model prediction is illustrated in \cref{fig:placeholder}. 
The RMSE of the difference between the signal and the model prediction in the test regime is $\eta = 0.9577$. The corresponding mean variance over samples is $0.00234$. 

\begin{figure}
    \centering
    \includegraphics[width=\linewidth]{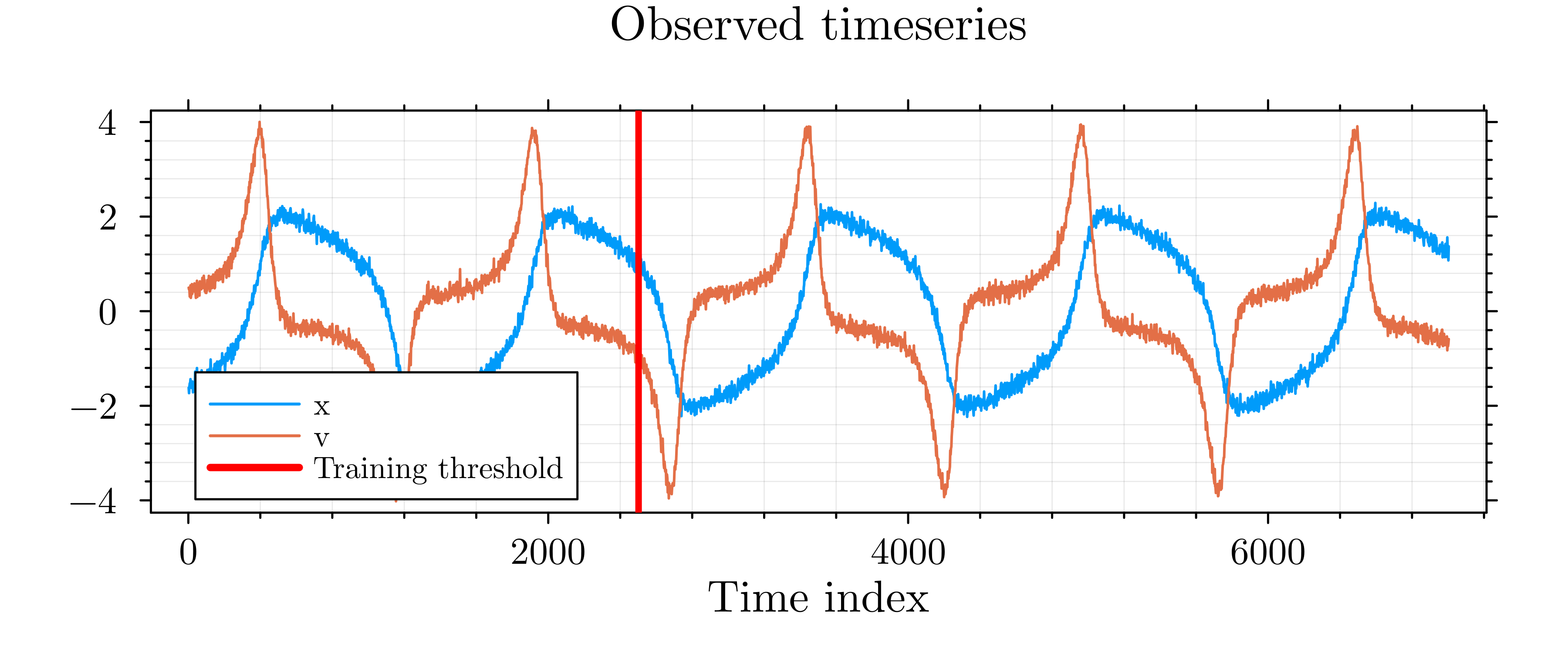}
    \includegraphics[width=\linewidth]{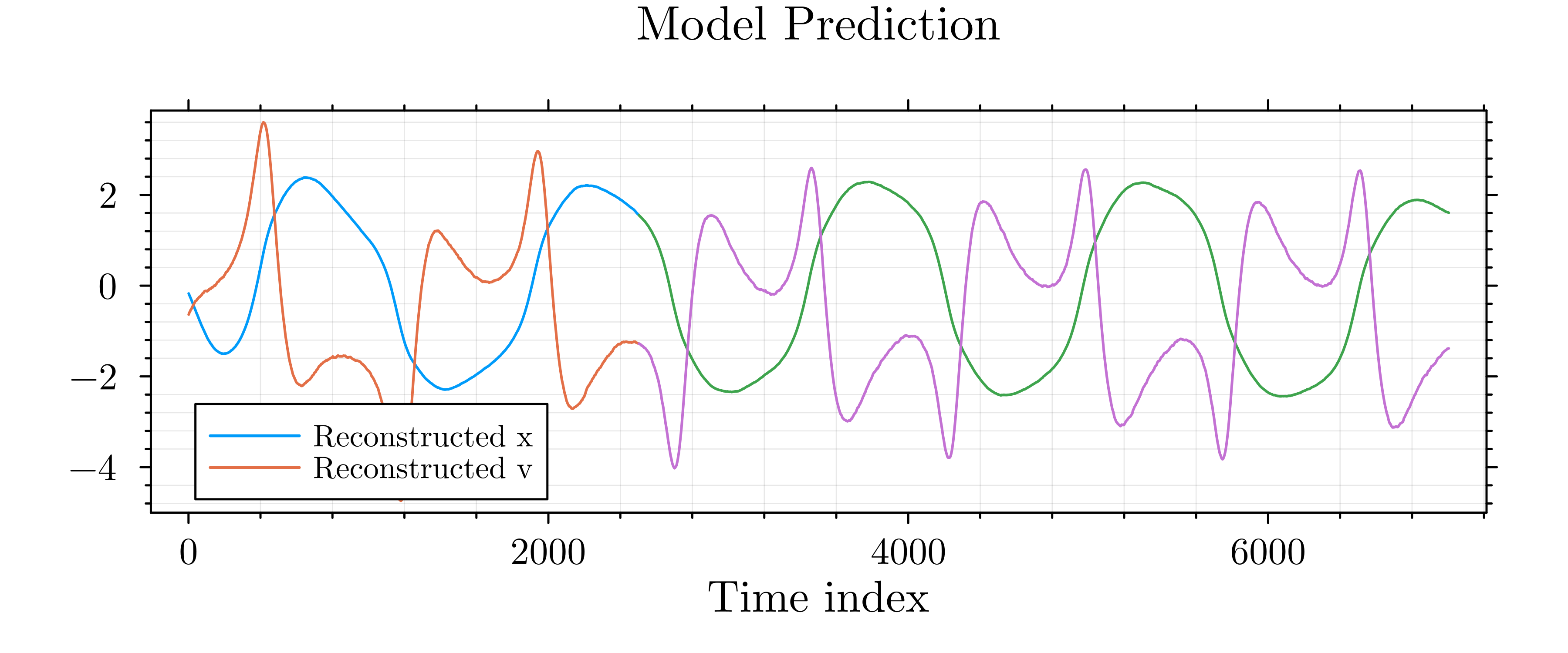}
    \caption{\textbf{Time-series modeling.}
     Here, time-index is the discrete sample number corresponding to uniformly sampled time points.
Top: Ground-truth trajectory.
Bottom: Model reconstruction. Blue and orange denote the training window; remaining segments are extrapolated.
The model preserves the limit-cycle dynamics outside the training domain.}
\label{fig:placeholder}
\end{figure}

\section{Conclusion}
This work has demonstrated how Barron functions and Wiener-Laguerre models may be integrated to learn continuous causal operators.
An interesting direction for future work is more robust uncertainty estimates by adapting $\sigma$ with data.
It is also of interest to extend the Wiener-Lagrange parameterization to generalized orthogonal basis filters.

\section*{Acknowledgment}
This work was partially supported by the Wallenberg AI, Autonomous Systems and
Software Program (WASP) funded by the Knut and Alice Wallenberg Foundation.

\onecolumn
\bibliographystyle{IEEEtran}
\bibliography{references}

\end{document}